\documentclass[runningheads]{llncs}
\usepackage{amsmath,amssymb,amsfonts}
\usepackage{xcolor}
\usepackage{tikz}
\usetikzlibrary{intersections}
\usepackage{thmtools}
\usepackage{thm-restate}

\makeatletter
\let\MYcaption\@makecaption
\makeatother
\usepackage{subcaption}
\usepackage{graphicx}
\usepackage{ifthen}

\allowdisplaybreaks[1]
\newcommand{\ov}[1]{\overline{#1}}

\newcommand{\eqdef}{\stackrel{\mbox{\scriptsize def}}{=}}

\newcommand{\reals}{\mathbb{R}}
\newcommand{\realsnng}{\mathbb{R}^{\ge0}}
\newcommand{\cala}{\mathcal{A}}

\newcommand{\cald}{\mathcal{D}}
\newcommand{\calf}{\mathcal{F}}

\newcommand{\calo}{\mathcal{O}}

\newcommand{\calr}{\mathcal{R}}
\newcommand{\cals}{\mathcal{S}}

\newcommand{\calw}{\mathcal{W}}

\newcommand{\supp}{{\mathtt{supp}}}
\newcommand{\Dists}{\mathbb{D}} 

\newcommand{\randassign}{\ensuremath{\stackrel{\mathrm{\$}}{\leftarrow}}}

\newcommand{\diverge}[2]{\mathit{D}(#1 \parallel #2)}

\newcommand{\renyidiverge}[3]{\mathit{D}_{#3}(#1 \parallel #2)}

\newcommand{\Dtv}[0]{\mathit{D}_{\sf tv}}
\newcommand{\tvdiverge}[2]{\renyidiverge{#1}{#2}{\sf tv}}

\newcommand{\Winfu}{\mathit{W}_{d}}

\newcommand{\StatEL}{\mbox{\rm StatEL}}
\renewcommand{\phi} {\varphi}
\newcommand{\erightarrow}{\supset}

\newcommand{\hy}{\hat{y}}

\newcommand{\hell}{\hat{\ell}}

\newcommand{\Var}{\mathtt{Mes}}
\newcommand{\Label}{\mathtt{L}}

\newcommand{\PR}[1]{\mathop{\mathbb{P}_{#1}}}
\newcommand{\Qw}[1]{\mathop{\mathbb{Q}_{#1}}}

\newcommand{\MKa}{\mathop{\mathsf{K}_{a}}}

\newcommand{\MKe}{\mathop{\mathsf{K}_{\varepsilon}}}
\newcommand{\MKeD}{\mathop{\mathsf{K}_{\varepsilon}^{\!D}}}

\newcommand{\MPa}{\mathop{\mathsf{P}_{\!a}}}

\newcommand{\MPe}{\mathop{\mathsf{P}_{\!\varepsilon}}}

\newcommand{\MPetv}{\mathop{\mathsf{P}_{\varepsilon}^{\sf tv}}}
\newcommand{\MPztv}{\mathop{\mathsf{P}_{0}^{\sf tv}}}
\newcommand{\MPelps}{\mathop{\mathsf{P}_{\varepsilon}^{r,D}}}

\newcommand{\M}{\mathfrak{M}}

\newcommand{\Ra}[0]{\calr_{a}}

\newcommand{\Reps}[0]{\calr_{\!\varepsilon}}
\newcommand{\Repsd}[0]{\calr_{\!\varepsilon}^{\!D}}
\newcommand{\Repstv}[0]{\calr_{\!\varepsilon}^{\sf tv}}

\newcommand{\Repslps}[0]{\calr_{\!\varepsilon}^{r,D}}

\newcommand{\wre}[0]{\mathit{w_{\sf real}}}

\newcommand{\TP}{\mathit{tp}}
\newcommand{\TN}{\mathit{tn}}
\newcommand{\FP}{\mathit{fp}}
\newcommand{\FN}{\mathit{fn}}

\newcommand{\Prevalence}{\mathsf{Prevalence}}
\newcommand{\Accuracy}{\mathsf{Accuracy}}
\newcommand{\Precision}{\mathsf{Precision}}
\newcommand{\FDR}{\mathsf{FDR}}
\newcommand{\FOR}{\mathsf{FOR}}
\newcommand{\NPV}{\mathsf{NPV}}
\newcommand{\Recall}{\mathsf{Recall}}
\newcommand{\FallOut}{\mathsf{FallOut}}
\newcommand{\MissRate}{\mathsf{MissRate}}
\newcommand{\Specificity}{\mathsf{Specificity}}

\newcommand{\TotalRobust}{\mathsf{TotalRobust}}
\newcommand{\TargetRobust}{\mathsf{TargetRobust}}

\newcommand{\GrpFair}{\mathsf{GrpFair}}
\newcommand{\IndFair}{\mathsf{IndFair}}
\newcommand{\EqOpp}{\mathsf{EqOpp}}

\newif\ifcommentson\commentsonfalse

\newif\ifconferenceon\conferenceonfalse
\ifconferenceon
\newcommand{\arxiv}[1]{}
\newcommand{\conference}[1]{#1}
\newcommand{\conferenceShort}[1]{}
\else
\newcommand{\arxiv}[1]{#1}
\newcommand{\conference}[1]{}
\newcommand{\conferenceShort}[1]{}
\fi

\newcommand{\journal}[1]{}

\ifcommentson
\newcommand{\commentsize}[0]{.90\textwidth}
\newcommand{\commentYK}[1]{\begin{center} \parbox{\commentsize}{\textbf{\textcolor{black}{Comment Y.}} \textcolor{red}{#1} }\end{center}}
\newcommand{\replyYK}[1]{\begin{center} \parbox{\commentsize}{\textbf{Reply Y.} \textcolor{blue}{#1} }\end{center}}
\newcommand{\commentY}[1]{\marginpar{\footnotesize \color{red} {\bf Y:} \textsf{\scriptsize #1}}}
\newcommand{\replyY}[1]{\marginpar{\footnotesize \color{red} {\bf Y:} \textsf{\scriptsize #1}}}
\else
\newcommand{\commentYK}[1]{}
\newcommand{\replyYK}[1]{}
\newcommand{\commentY}[1]{}
\newcommand{\replyY}[1]{}
\fi

\newcommand{\colorR}[1]{\textcolor{red}{#1}}

\newcommand{\pagelimitmarker}[1]{~\\ {\colorR{\ifthenelse{\thepage>#1}{\Huge Exceeding the page limit}{\huge Within the page limit}}}~\\ {\huge{\colorR{~~Page Limit\,\,\,\,\, = #1}}}~\\ {\huge{\colorR{~~Current Page = $\thepage$}}}}

\begin{document}
\title{Towards Logical Specification of Statistical Machine Learning
\thanks{This work was supported by JSPS KAKENHI Grant Number JP17K12667, by the New Energy and Industrial Technology Development Organization (NEDO), and by Inria under the project LOGIS.}}
\author{Yusuke Kawamoto\inst{1}\orcidID{0000-0002-2151-9560}}
\authorrunning{Y. Kawamoto}
\institute{AIST, Tsukuba, Japan \\~
\conference{\email{yusuke.kawamoto.aist@gmail.com}}
}
\maketitle              
\begin{abstract}
We introduce a logical approach to formalizing statistical properties of machine learning.
Specifically, we propose a formal model for statistical classification based on a Kripke model, and formalize various notions of classification performance, robustness, and fairness of classifiers by using epistemic logic.
Then we show some relationships among properties of classifiers and those between classification performance and robustness, which suggests robustness-related properties that have not been formalized in the literature as far as we know.
To formalize fairness properties, we define a notion of counterfactual knowledge and show techniques to formalize conditional indistinguishability by using counterfactual epistemic operators.
As far as we know, this is the first work that uses logical formulas to express statistical properties of machine learning, and that provides epistemic (resp. counterfactually epistemic) views on robustness (resp. fairness) of classifiers.

\end{abstract}

\section{Introduction}
\label{sec:intro}
With the increasing use of machine learning in real-life applications, 
the safety and security of learning-based systems have been of great interest.
In particular, many recent studies~\cite{Szegedy:14:ICLR,Chakraborty:18:arxiv} have found vulnerabilities on the robustness of deep neural networks (DNNs) to malicious inputs, which can lead to disasters in security critical systems, such as self-driving cars.
To find out these vulnerabilities in advance, there have been researches on the formal verification and testing methods for the robustness of DNNs in recent years~\cite{Huang:17:CAV,Katz:17:CAV,Pei:17:SOSP,Tian:18:ICSE}.
However, relatively little attention has been paid to the formal specification of machine learning~\cite{Seshia:18:ATVA}.

To describe the formal specification of security properties, logical approaches have been shown useful to classify desired properties and to develop theories to compare those properties.
For example, security policies in temporal systems have been formalized as trace properties~\cite{Alpern:85:IPL} or hyperproperties~\cite{Clarkson:08:CSF}, which characterize the relationships among various security policies.
For another example, epistemic logic~\cite{vonWright:51:book} has been widely used as formal policy languages (e.g., for the authentication~\cite{Burrows:90:TOCS} and the anonymity~\cite{Syverson:99:FM,Halpern:05:JCS} of security protocols, and for the privacy of social network~\cite{Pardo:14:SEFM}).
As far as we know, however, no prior work has employed logical formulas to rigorously describe various statistical properties of machine learning, although there are some papers that (often informally) list various desirable properties of machine learning~\cite{Seshia:18:ATVA}.

In this paper, we present a first logical formalization of statistical properties of machine learning.
To describe the statistical properties in a simple and abstract way, we employ \emph{statistical epistemic logic} (\StatEL{})~\cite{Kawamoto:19:FC:arxiv}, which is recently proposed to describe statistical knowledge and is applied to formalize statistical hypothesis testing and statistical privacy of databases.

A key idea in our modeling of statistical machine learning is that we formalize logical properties in the syntax level by using logical formulas, and statistical distances in the semantics level by using accessibility relations of a Kripke model~\cite{Kripke:63:MLQ}.
In this model, we formalize statistical classifiers and some of their desirable properties: classification performance, robustness, and fairness.
More specifically, classification performance and robustness are described as the differences between the classifier's recognition and the correct label (e.g., given by the human), whereas fairness is formalized as the conditional indistinguishability between two groups or individuals by using a notion of counterfactual knowledge.

\paragraph{Our contributions.}
The main contributions of this work are as follows:
\begin{itemize}
\item We show a logical approach to formalizing statistical properties of machine learning in a simple and abstract way.
In particular, we model logical properties in the syntax level, and statistical distances in the semantics level.
\item We introduce a formal model for statistical classification.
More specifically, we show how probabilistic behaviors of classifiers and non-deterministic adversarial inputs are formalized in a distributional Kripke model~\cite{Kawamoto:19:FC:arxiv}.
\item We formalize the classification performance, robustness, and fairness of classifiers by using statistical epistemic logic (\StatEL{}).
As far as we know, this is the first work that uses logical formulas to formalize various statistical properties of machine learning, and that provides epistemic (resp. counterfactually epistemic) views on robustness (resp. fairness) of classifiers.
\item We show some relationships among properties of classifiers, e.g., different strengths of robustness.
We also present some relationships between classification performance and robustness, which suggest robustness-related properties that have not been formalized in the literature as far as we know.
\item To formalize fairness properties, we define a notion of certain counterfactual knowledge and show techniques to formalize conditional indistinguishability by using counterfactual epistemic operators in \StatEL{}.
This enables us to express various fairness properties in a similar style of logical formulas.
\end{itemize}

\paragraph{Cautions and limitations.}
In this paper, we focus on formalizing properties of classification problems and do not deal with the properties of learning algorithms (e.g., fairness through unawareness of sensitive attributes in data preparation), quality of training data (e.g., sample bias), quality of testing (e.g., coverage criteria), explainability, temporal properties, system level specification, 
or process agility in system development.
It should be noted that most of the properties formalized in this paper have been known in literatures on machine learning, and the novelty of this work lies in the logical formulation of those statistical properties.

We also remark that this work does not provide methods for checking, guaranteeing, or improving the performance/robustness/fairness of machine learning.
As for the satisfiability of logical formulas,
we leave the development of testing and (statistical) model checking algorithms as future work, since the research area on the testing and formal/statistical verification of machine learning is relatively new and needs further techniques to improve the scalability.
Moreover, in some applications such as image recognition, some formulas (e.g., representing whether an input image is panda or not) cannot be implemented mathematically, and require additional techniques based on experiments.
Nevertheless, we demonstrate that describing various properties using logical formulas is useful to explore desirable properties and to discuss their relationships in a framework.

Finally, we emphasize that our work is the first attempt to use logical formulas to express statistical properties of machine learning, and would be a starting point to develop theories of specification of machine learning in future research.

\paragraph*{Paper organization.}
The rest of this paper is organized as follows.
Section~\ref{sec:preliminaries} presents background on statistical epistemic logic (\StatEL{}) and notations used in this paper.
Section~\ref{sec:counterfactual} defines counterfactual epistemic operators and shows techniques to model conditional indistinguishability using \StatEL{}.
Section~\ref{sec:formal:ML} introduces a formal model for describing the behaviours of statistical classifiers and non-deterministic adversarial inputs.
Sections~\ref{sec:ML:performance-prediction}, \ref{sec:ML:robustness}, and~\ref{secML:fairness} respectively formalize the classification performance, robustness, and fairness of classifiers by using \StatEL{}.
Section~\ref{sec:related} presents related work and Section~\ref{sec:conclude} concludes.

\section{Preliminaries}
\label{sec:preliminaries}
In this section we introduce some notations and recall the syntax and semantics of the \emph{statistical epistemic logic} (\StatEL{}) introduced in~\cite{Kawamoto:19:FC:arxiv}.

\subsection{Notations}
\label{subsec:notations-pd}

Let $\realsnng$ be the set of non-negative real numbers,
and $[0, 1]$ be the set of non-negative real numbers not greater than~$1$.
We denote by $\Dists\calo$ the set of all probability distributions over a set~$\calo$. Given a finite set $\calo$ and a probability distribution $\mu\in\Dists\calo$, the probability of sampling a value $y$ from $\mu$ is denoted by $\mu[y]$.
For a subset $R\subseteq\calo$ we define $\mu[R]$ by: $\mu[R] = \sum_{y\in R} \mu[y]$.
For a distribution $\mu$ over a finite set $\calo$, its \emph{support} is defined by 
$\supp(\mu) = \{ v \in \calo \colon \mu[v] > 0 \}$.

The \emph{total variation distance} of two distributions $\mu, \mu' \in \Dists\calo$ is defined by:
$
\tvdiverge{\mu}{\mu'} \eqdef
\sup_{R \subseteq \calo} | \mu(R) - \mu'(R) |
$
{.}

\subsection{Syntax of \StatEL{}}
\label{sub:syntax}

We recall the syntax of the statistical epistemic logic (\StatEL{})~\cite{Kawamoto:19:FC:arxiv}, 
which has two levels of formulas: \emph{static} and \emph{epistemic formulas}.
Intuitively, a static formula describes a proposition satisfied at a deterministic state, while an epistemic formula describes a proposition satisfied at a probability distribution of states.
In this paper, the former is used only to define the latter.

Formally, let $\Var$ be a set of symbols called \emph{measurement variables}, and
$\Gamma$ be a set of atomic formulas of the form $\gamma(x_1, x_2, \ldots, x_n)$ for a predicate symbol $\gamma$, $n \ge 0$, and $x_1, x_2, \ldots, x_n\in\Var$.
Let $I \subseteq [0, 1]$ be a finite union of disjoint intervals, and $\cala$ be a finite set of indices (e.g., associated with statistical divergences).
Then the formulas are defined by:
\begin{itemize}
\item[] Static formulas:~~
$\psi \mathbin{::=}
 \gamma(x_1, x_2, \ldots, x_n) \mid
 \neg \psi \mid \psi \wedge \psi$
\item[] Epistemic formulas:~~
$\phi \mathbin{::=}
 \PR{I} \psi \mid \neg \phi \mid \phi \wedge \phi \mid
 \psi \erightarrow \phi \mid \MKa \phi$
\end{itemize}
where 
$a\in\cala$.
We denote by $\calf$ the set of all epistemic formulas.
Note that we have no quantifiers over measurement variables. 
(See Section~\ref{sub:interpretation} for more details.)

The \emph{probability quantification} $\PR{I} \psi$ represents that a static formula $\psi$ is satisfied with a probability belonging to a set $I$.
For instance, $\PR{(0.95, 1]} \psi$ represents that $\psi$ holds with a probability greater than $0.95$.
By $\psi \erightarrow \PR{I} \psi'$ we represent that the conditional probability of $\psi'$ given $\psi$ is included in a set $I$.
The \emph{epistemic knowledge} $\MKa \phi$ expresses that we knows $\phi$ with a confidence specified by~$a$.

As syntax sugar, we use \emph{disjunction} $\vee$, \emph{classical implication} $\rightarrow$, and \emph{epistemic possibility} $\MPa$, defined as usual by:
$\phi_0 \vee \phi_1 \mathbin{::=} \neg (\neg \phi_0 \wedge \neg \phi_1)$,
$\phi_0 \rightarrow \phi_1 \mathbin{::=} \neg \phi_0 \vee \phi_1$,
and $\MPa{\phi} \mathbin{::=} \neg \MKa \neg \phi$.
When $I$ is a singleton $\{ i \}$, we abbreviate $\PR{I}$ as $\PR{i}$.

\subsection{Distributional Kripke Model}
\label{sub:Kripke}

Next we recall the notion of a distributional Kripke model~\cite{Kawamoto:19:FC:arxiv}, where each possible world is a probability distribution over a set $\cals$ of states and
each world $w$ is associated with a stochastic assignment $\sigma_w$ to measurement variables.

\begin{definition}[Distributional Kripke model] \label{def:dist-Kripke-model} \rm
Let $\cala$ be a finite set of indices (typically associated with statistical tests and their thresholds),
$\cals$ be a finite set of states, and $\calo$ be a finite set of data.
A \emph{distributional Kripke model} 
is a tuple 
$\M =(\calw, (\calr_a)_{a\in\cala}, (V_s)_{s\in\cals})$ 
consisting of:
\begin{itemize}
\item a non-empty set $\calw$ of probability distributions over a finite set $\cals$ of states;
\item for each $a\in\cala$, an accessibility relation $\calr_a \subseteq \calw \times \calw$;
\item for each $s\in\cals$, a valuation $V_s$ that maps each $k$-ary predicate $\gamma$ to a set $V_s(\gamma) \subseteq \calo^k$.
\end{itemize}
The set $\calw$ is called a \emph{universe}, and its elements are called \emph{possible worlds}.
All measurement variables range over the same set $\calo$ in every world.

We assume that each $w\in\calw$ is associated with a function $\rho_w: \Var\times\cals\rightarrow\calo$ that maps each measurement variable $x$ to its value $\rho_w(x, s)$ observed at a state~$s$.
We also assume that each state $s$ in a world $w$ is associated with the assignment $\sigma_s: \Var\rightarrow\calo$ defined by $\sigma_s(x) = \rho_w(x, s)$.
\end{definition}

Since each world $w$ is a distribution of states, we denote by $w[s]$ the probability that a state $s$ is sampled from $w$.
Then the probability that a measurement variable $x$ has a value $v$ is given by
$\sigma_w(x)[v] = 
\sum_{\substack{s\in\supp(w), \sigma_s(x) = v}} w[s]
$.
This implies that, when a state $s$ is drawn from $w$, an input $\sigma_s(x)$ is sampled from the distribution $\sigma_w(x)$.

\subsection{Stochastic Semantics of \StatEL{}}
\label{sub:interpretation}

Now we recall the \emph{stochastic semantics}~\cite{Kawamoto:19:FC:arxiv} for the \StatEL{} formulas over a distributional Kripke model 
$\M =(\calw, (\calr_a)_{a\in\cala}, (V_s)_{s\in\cals})$ 
with $\calw = \Dists\cals$.

The interpretation of static formulas $\psi$ at a state $s$ is given by:
\begin{align*}
s \models \gamma(x_1, x_2, \ldots, x_k) 
& ~\mbox{ iff }~
(\sigma_s(x_1), \sigma_s(x_2), \ldots, \sigma_s(x_k)) \in V_s(\gamma)
\\
s \models \neg \psi 
& ~\mbox{ iff }~
s \not\models \psi 
\\
s \models \psi \wedge \psi'
& ~\mbox{ iff }~
s \models \psi
~\mbox{ and }~
s \models \psi'
{.}
\end{align*}

The \emph{restriction} $w|_\psi$ of a world $w$ to a static formula $\psi$ is defined by
$w|_\psi[s] = \frac{w[s]}{\sum_{s': s' \models \psi} w[s']}$ if $s \models \psi$,
and $w|_\psi[s] = 0$ otherwise.
Note that $w|_\psi$ is undefined if there is no state $s$ that satisfies $\psi$ and has a non-zero probability in $w$.

Then the interpretation of epistemic formulas in a world $w$ is defined by:
\begin{align*}
\M, w \models \PR{I} \psi
& ~\mbox{ iff }~
\Pr\!\left[ s \randassign w :~ s \models \psi \right] \in I
\\
\M, w \models \neg \phi
& ~\mbox{ iff }~
\M, w \not\models \phi
\\
\M, w \models \phi \wedge \phi'
& ~\mbox{ iff }~
\M, w \models \phi
~\mbox{ and }~
\M, w \models \phi'
\\
\M, w \models \psi \erightarrow \phi
& ~\mbox{ iff }~
\mbox{$w|_{\psi}$ is defined and }~
\M, w|_{\psi} \models \phi
\\
\M, w \models \MKa \phi
& ~\mbox{ iff }~
\mbox{for every $w'$ s.t. $(w, w') \in \calr_a$, }~
\M, w' \models \phi
{,}
\end{align*}
where $s \randassign w$ represents that a state $s$ is sampled from the distribution $w$.

Then $\M, w \models \psi_0 \erightarrow \PR{I} \psi_1$ represents that 
the conditional probability of satisfying a static formula $\psi_1$ given another $\psi_0$ is included in a set $I$ at a world~$w$.

In each world $w$, measurement variables can be interpreted using $\sigma_w$. 
This allows us to assign different values to different occurrences of a variable in a formula;
E.g., in $\phi(x) \rightarrow \MKa \phi'(x)$,\, $x$ occurring in $\phi(x)$ is interpreted by $\sigma_{w}$ in a world $w$, while $x$ in $\phi'(x)$ is interpreted by $\sigma_{w'}$ in another $w'$ s.t. $(w, w')\in \calr_a$.

Finally, the interpretation of an epistemic formula $\phi$ in $\M$ is given by:
\begin{align*}
\M \models \phi
& ~\mbox{ iff }~
\mbox{for every world $w$ in $\M$, }~
\M, w \models \phi
{.}
\end{align*}

\section{Techniques for Conditional Indistinguishability}
\label{sec:counterfactual}
In this section we introduce some modal operators to define a notion of ``counterfactual knowledge'' using \StatEL{}, and show how to employ them to formalize conditional indistinguishability properties.
The techniques presented here are used to formalize some fairness properties of machine learning in Section~\ref{secML:fairness}.

\subsection{Counterfactual Epistemic Operators}
\label{sub:counterfactual-epistemic}

Let us consider an accessibility relation $\Reps$ based on a statistical divergence $\diverge{\cdot}{\cdot}: \Dists\calo\times\Dists\calo \rightarrow \realsnng$ and a threshold $\varepsilon\in\realsnng$ defined by:
\[
\Reps \eqdef
\left\{ (w, w') \in \calw\times\calw \mid D(\sigma_{w}(y) \parallel \sigma_{w'}(y)) \le \varepsilon
\right\}
{,}
\]
where $y$ is the measurement variable observable in each world in $\calw$.
Intuitively, $(w, w')\in\Reps$ represents that the probability distribution $\sigma_{w}(y)$ of the data $y$ observed in a world $w$ is \emph{indistinguishable} from that in another world $w'$ in terms of $D$.

Now we define the complement relation of $\Reps$ by $\ov{\Reps} \eqdef (\calw\times\calw) \setminus \Reps$, namely,
\[
\ov{\Reps} =
\left\{ (w, w') \in \calw\times\calw \mid D(\sigma_{w}(y) \parallel \sigma_{w'}(y)) > \varepsilon
\right\}
{.}
\]
Then $(w, w')\in\ov{\Reps}$ represents that the distribution $\sigma_{w}(y)$ observed in $w$ \emph{can be distinguished} from that in $w'$.
Then the corresponding epistemic operator $\ov{\MKe}$, which we call a \emph{counterfactual epistemic operator}, is interpreted as:
\begin{align}
\label{eq:counterfactual1}
\M, w \models \ov{\MKe} \phi
& ~\mbox{ iff }~
\mbox{for every $w'$ s.t. }
(w, w') \in \ov{\Reps},
\mbox{ we have }
\M, w' \models \phi
\\
\label{eq:counterfactual2}
& ~\mbox{ iff }~
\mbox{for every $w'$ s.t. }
\M, w' \models \neg\phi,
\mbox{ we have }
(w, w') \in \Reps
{.}
\end{align}
Intuitively, \eqref{eq:counterfactual1} represents that if we were located in a possible world $w'$ that looked distinguished from the real world $w$, then $\phi$ would always hold.
This means a \emph{counterfactual knowledge}\footnote{Our definition of counterfactual knowledge is limited to the  condition of having an observation different from the actual one.
More general notions of counterfactual knowledge can be found in previous work (e.g.,~\cite{Williamson:07:GPS}).} in the sense that, if we had an observation different from the real world, then we would know $\phi$.
This is logically equivalent to \eqref{eq:counterfactual2}, representing that all possible worlds $w'$ that do not satisfy $\phi$ look indistinguishable from the real world $w$ in terms of $D$.

We remark that the dual operator $\ov{\MPe}$ is interpreted as:
\begin{align}
\label{eq:counterfactual3}
\M, w \models \ov{\MPe} \phi
& ~\mbox{ iff }~
\mbox{there exists a $w'$ s.t. }
(w, w') \notin \Reps
\mbox{ and }
\M, w' \models \phi
{.}
\end{align}
This means a counterfactual possibility in the sense that it might be the case where we had an observation different from the real world and thought $\phi$ possible.

\subsection{Conditional Indistinguishability via Counterfactual Knowledge}
\label{sub:CondIndCounterfactual}

As shown in Section~\ref{secML:fairness}, some fairness notions in machine learning are based on conditional indistinguishability of the form~\eqref{eq:counterfactual2}, hence can be expressed using counterfactual epistemic operators.

Specifically, we use the following proposition, stating that given that two static formulas $\psi$ and $\psi'$ are respectively satisfied in worlds $w$ and $w'$ with probability $1$, then the indistinguishability between $w$ and $w'$ can be expressed as
$w \models \psi \erightarrow \neg \ov{\MPa }\PR{1} \psi'$.
Note that this formula means that there is no possible world where we have an observation different from the real world $w$ (satisfying $\psi$) but we think $\psi'$ possible;
i.e., the formula means that if $\psi'$ is satisfied then we have an observation indistinguishable from that in the real world $w$.

\begin{restatable}[Conditional indistinguishability]{prop}{CondInd}
\label{prop:CondInd}
Let $\M =(\calw, (\calr_a)_{a\in\cala}, \allowbreak (V_s)_{s\in\cals})$ 
be a distributional Kripke model with the universe $\calw = \Dists\cals$.
Let $\psi$ and $\psi'$ be static formulas, and $a\in\cala$.
\begin{enumerate}
\item[(i)] 
$\M \models \psi \erightarrow \neg \ov{\MPa }\PR{1} \psi'$
iff for any $w, w' \in \calw$,\, 
$\M, w \models \PR{1} \psi$ and $\M, w' \models \PR{1} \psi'$
imply $(w, w') \in \Ra$.
\item[(ii)] 
If $\calr_a$ is symmetric, then $\M \models \psi \erightarrow \neg \ov{\MPa }\PR{1} \psi'$ iff $\M \models \psi' \erightarrow \neg \ov{\MPa }\PR{1} \psi$.
\end{enumerate}
\end{restatable}

See Appendix~\ref{sec:proof:prop:CondInd} for the proof.

\section{Formal Model for Statistical Classification}
\label{sec:formal:ML}
In this section we introduce a formal model for statistical classification by using distributional Kripke models (Definition~\ref{def:dist-Kripke-model}).
In particular, we formalize a probabilistic behaviour of a classifier $C$ and a non-deterministic input $x$ from an adversary in a distributional Kripke model.

\subsection{Statistical Classification Problems}
\label{sub:ML:classifications}

\emph{Multiclass classification} is the problem of classifying a given input into one of multiple classes.
Let $\Label$ be a finite set of \emph{class labels}, and $\cald$ be a finite set of \emph{input data} (called \emph{feature vectors}) that we want to classify.
Then a \emph{classifier} is a function $C: \cald\rightarrow\Label$ that receives an input  datum and predicts which class (among $\Label$) the input belongs to.
Here we do \emph{not} model how classifiers are constructed from a set of training data, but deal with a situation where some classifier $C$ has already been obtained and its properties should be evaluated.

Let $f: \cald\times\Label \rightarrow \reals$ be a \emph{scoring function} that gives a score $f(v, \ell)$ of predicting the class of an input datum (feature vector) $v$ as a label $\ell$.
Then for each input $v\in\cald$, we denote by $H(v) = \ell$ to represent that a label $\ell$ maximizes $f(v, \ell)$.
For example, when the input $v$ is an image of an animal and $\ell$ is the animal's name, $H(v) = \ell$ may represent that an oracle (or ``human'') classifies the image $v$ as~$\ell$.

\subsection{Modeling the Behaviours of Classifiers}
\label{sub:ML:predicates}

Classifiers are formalized on a distributional Kripke model 
$\M =(\calw, (\calr_a)_{a\in\cala}, \allowbreak (V_s)_{s\in\cals})$ 
with $\calw = \Dists\cals$ and a real world $\wre \in \calw$.
Recall that each world $w\in\calw$ is a probability distribution over the set $\cals$ of states and has a stochastic assignment $\sigma_w: \Var \rightarrow \Dists\calo$ that is consistent with the deterministic assignments~$\sigma_s$ for all $s\in\cals$ (as explained in Section~\ref{sub:Kripke}).

We present an overview of our formalization in Fig.~\ref{fig:states}.
We denote by $x\in\Var$ an input datum given to the classifier $C$ (and to the oracle $H$), by $y\in\Var$ a correct label given by the oracle $H$, and by $\hy\in\Var$ a label predicted by $C$.
We assume that the input variable $x$ (resp. the output variables $y,\hy$) ranges over the set $\cald$ of input data (resp. the set $\Label$ of labels);
i.e., the deterministic assignment $\sigma_s$ at each state $s\in\cals$ has the range $\calo = \cald \cup \Label$ and satisfies $\sigma_s(x)\in\cald$ and $\sigma_s(y), \sigma_s(\hy)\in\Label$.

A key idea in our modeling is that we formalize logical properties in the syntax level by using logical formulas, and statistical distances in the semantics level by using accessibility relations $\calr_a$.
In this way, we can formalize various statistical properties of classifiers in a simple and abstract way.

To formalize a classifier $C$, we introduce a static formula $\psi(x, \hy)$ to represent that $C$ classifies a given input $x$ as a class $\hy$.
We also introduce a static formula $h(x, y)$ to represent that $y$ is the actual class of an input $x$.
As an abbreviation, we write $\psi_\ell(x)$ (resp. $h_\ell(x)$) to denote $\psi(x, \ell)$ (resp. $h(x, \ell)$).
Formally, these static formulas are interpreted at each state $s\in\cals$ as follows:
\begin{align*}
s \models \psi(x, \hy) &~\mbox{ iff }~
C(\sigma_s(x)) = \sigma_s(\hy).
\\
s \models h(x, y) &~\mbox{ iff }~
H(\sigma_s(x)) = \sigma_s(y).
\end{align*}

\begin{figure}[t]
\centering
\begin{picture}(290, 127)
 \put(50,94){\scriptsize State $s_0$}
 \put(172,95){\oval(170,35)}
 \put(155,85){\framebox(35,20)}
 \put(131,105){\tiny input}
 \put(105,94){\scriptsize $\sigma_{s_0}(x)$}
 \put(131,95){\vector(1,0){20}}
 \put(196,105){\tiny output}
 \put(220,94){\scriptsize $\sigma_{s_0}(\hy)$}
 \put(195,95){\vector(1,0){20}}
 \put(159,99){\tiny Classifier}
 \put(168,89){$C$}
 \put(50,44){\scriptsize State $s_1$}
 \put(172,45){\oval(170,35)}
 \put(155,35){\framebox(35,20)}
 \put(131,55){\tiny input}
 \put(105,44){\scriptsize $\sigma_{s_1}(x)$}
 \put(131,45){\vector(1,0){20}}
 \put(196,55){\tiny output}
 \put(220,44){\scriptsize $\sigma_{s_1}(\hy)$}
 \put(195,45){\vector(1,0){20}}
 \put(159,49){\tiny Classifier}
 \put(168,39){$C$}
 \put(60,15){\rotatebox{90}{$\cdots$}}
 \put(170,8){\rotatebox{90}{\footnotesize $\cdots$}}
 \put(0,74){\scriptsize World $w$}
 \linethickness{0.1pt}
 \put(40,1){\dashbox{2}(235,125)}
\end{picture}
\caption{
A world $w$ is chosen non-deterministically.
With probability $w[s_i]$, the world $w$ is in a deterministic state $s_i$ where the classifier $C$ receives the input value $\sigma_{s_i}(x)$ and returns the output value $\sigma_{s_i}(\hy)$.
\label{fig:states}}
\end{figure}
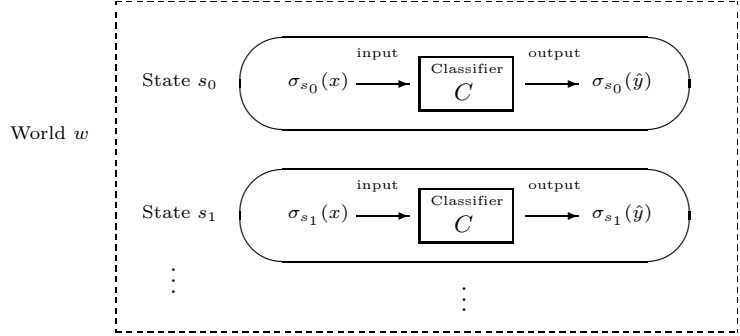

\subsection{Modeling the Non-deterministic Inputs from Adversaries}
\label{sub:ML:non-deterministic-inputs}

As explained in Section~\ref{sub:Kripke},
when a state $s$ is drawn from a distribution $w\in\calw$, an input value $\sigma_s(x)$ is sampled from the distribution $\sigma_w(x)$, and assigned to the measurement variable~$x$.
Since $x$ denotes the input to the classifier $C$, the input distribution $\sigma_w(x)$ over $\cald$ can be regarded as the \emph{test dataset}.
This means that each world $w$ corresponds to a test dataset $\sigma_w(x)$.
For instance, $\sigma_{\wre}(x)$ in the real world $\wre$ represents the actual test dataset.
The set of all possible test datasets (i.e., possible distributions of inputs to $C$) is represented by
$\Lambda \eqdef \left\{ \sigma_w(x) \mid w\in\calw \right\}$.
Note that $\Lambda$ can be an infinite set.

For example, let us consider testing the classifier $C$ with the actual test dataset $\sigma_{\wre}(x)$.
When $C$ assigns a label $\ell$ to an input $x$ with probability $0.2$, i.e.,
$
\Pr\!\left[~ v \randassign \sigma_{\wre}(x) \,:\, 
C(v) = \ell ~\right] = 0.2
$,
then this can be expressed by:
\begin{align*}
\M, \wre \models \PR{0.2} \psi_\ell(x)
{.}
\end{align*}

We can also formalize a non-deterministic input $x$ from an adversary in this model as follows.
Although each state $s$ in a possible world $w$ is assigned the probability $w[s]$, each possible world $w$ itself is not assigned a probability.
Thus, each input distribution $\sigma_w(x) \in \Lambda$ itself is also not assigned a probability, hence our model assumes no probability distribution over $\Lambda$.
In other words, we assume that a world $w$ and thus an adversary's input distribution $\sigma_w(x)$ are non-deterministically chosen.
This is useful to model an adversary's malicious inputs in the definitions of security properties, because we usually do not have a prior knowledge of the distribution of malicious inputs from adversaries, and need to reason about the worst cases caused by the attack.
In Section~\ref{sec:ML:robustness}, this formalization of non-deterministic inputs is used to express the robustness of classifiers.

Finally, it should be noted that we cannot enumerate all possible adversarial inputs, hence cannot construct $\calw$ by collecting their corresponding worlds.
Since $\calw$ can be an infinite set and is unspecified, we do not aim at checking whether or not a formula is satisfied in all possible worlds of $\calw$.
Nevertheless, as shown in later sections, describing various properties using \StatEL{} is useful to explore desirable properties and to discuss relationships among them.

\section{Formalizing the Classification Performance}
\label{sec:ML:performance-prediction}

In this section we show a formalization of classification performance using \StatEL{} (See Fig.~\ref{fig:performance-robustness} for basic ideas).
In classification problems, the terms \emph{positive}/\emph{negative} represent the result of the classifier's prediction, and the terms \emph{true}/\emph{false} represent whether the classifier predicts correctly or not.
Then the following terminologies are commonly used:
\begin{itemize}
\item[($\TP$)\!] \emph{true positive} means both the prediction and actual class are positive;
\item[($\TN$)\!] \emph{true negative} means both the prediction and actual class are negative;
\item[($\FP$)\!] \emph{false positive} means the prediction is positive but the actual class is negative;
\item[($\FN$)\!] \emph{false negative} means the prediction is negative but the actual class is positive.
\end{itemize}
These terminologies can be formalized using \StatEL{} as shown in Table~\ref{table:confusion}.
For example, when an input $x$ shows true positive at a state $s$, this can be expressed as $s \models \psi_\ell(x) \wedge h_\ell(x)$.
True negative, false positive (Type I error), and false negative (Type II error) are respectively expressed as 
$s \models \neg \psi_\ell(x) \wedge \neg h_\ell(x)$,\, 
$s \models \psi_\ell(x) \wedge \neg h_\ell(x)$, and 
$s \models \neg \psi_\ell(x) \wedge h_\ell(x)$.

\begin{table}[t]
  \caption{Logical description of the table of confusion \label{table:confusion}}
  \centering
  \scalebox{0.9}{
  \begin{tabular}{|l|l|l|l|l|} \cline{1-3}
    \multicolumn{1}{|c|}{}& \multicolumn{2}{c|}{Actual class} \\ \cline{2-5}
    \multicolumn{1}{|c|}{}&
    \multicolumn{1}{c|}{positive} & \multicolumn{1}{c|}{negative}
    & ~$\Prevalence_{\ell,I}(x) \eqdef$
    & ~$\Accuracy_{\ell,I}(x) \eqdef$ \\
    \multicolumn{1}{|c|}{}&
    \multicolumn{1}{c|}{$h_\ell(x)$} & \multicolumn{1}{c|}{$\neg h_\ell(x)$}
    & ~$\PR{I}(h_\ell(x))$
    & ~$\PR{I}(\psi_\ell(x) \leftrightarrow h_\ell(x))$\! \\[0.5ex] \hline
    Positive & & & & \\[-1.1ex]
    prediction & ~$\TP(x) \eqdef$ & ~$\FP(x) \eqdef$ 
    & ~$\Precision_{\ell,I}(x) \eqdef$
    & ~$\FDR_{\ell,I}(x) \eqdef$ \\
    $\psi_\ell(x)$ &
    ~$\psi_\ell(x) \wedge h_\ell(x)$ & ~$\psi_\ell(x) \wedge \neg h_\ell(x)$
    & ~$\psi_\ell(x) \erightarrow \PR{I} h_\ell(x)$
    & ~$\psi_\ell(x) \erightarrow \PR{I} \neg h_\ell(x)$ \\ \cline{1-5}
    Negative & & & & \\[-1.1ex]
    prediction & ~$\FN(x) \eqdef$ & ~$\TN(x) \eqdef$ 
    & ~$\FOR_{\ell,I}(x) \eqdef$
    & ~$\NPV_{\ell,I}(x) \eqdef$ \\
    $\neg\psi_\ell(x)$ &
    ~$\neg \psi_\ell(x) \wedge h_\ell(x)$ & ~$\neg \psi_\ell(x) \wedge \neg h_\ell(x)$
    & ~$\neg \psi_\ell(x) \erightarrow \PR{I} h_\ell(x)$
    & ~$\neg \psi_\ell(x) \erightarrow \PR{I} \neg h_\ell(x)$ \\ \hline
    \multicolumn{1}{c|}{} & ~$\Recall_{\ell,I}(x) \eqdef$~ 
    & ~$\FallOut_{\ell,I}(x) \eqdef$~ \\
    \multicolumn{1}{c|}{} & ~$h_\ell(x) \erightarrow \PR{I} \psi_\ell(x)$~~
    & ~$\neg h_\ell(x) \erightarrow \PR{I} \psi_\ell(x)$~~ \\ \cline{2-3}
    \multicolumn{1}{c|}{} & ~$\MissRate_{\ell,I}(x) \eqdef$~ 
    & ~$\Specificity_{\ell,I}(x) \eqdef$~ \\
    \multicolumn{1}{c|}{} & ~$h_\ell(x) \erightarrow \PR{I} \neg\psi_\ell(x)$
    & ~$\neg h_\ell(x) \erightarrow \PR{I} \neg\psi_\ell(x)$ \\ \cline{2-3}
  \end{tabular}
}
\end{table}

Then \emph{precision} (\emph{positive predictive value}) is defined as the conditional probability that the prediction is correct given that the prediction is positive; i.e., 
${\it precision} = \frac{\TP}{\TP + \FP}$.
Since the test dataset distribution in the real world $\wre$ is expressed as $\sigma_{\wre}(x)$ (as explained in Section~\ref{sub:ML:non-deterministic-inputs}),
the precision being within an interval $I$ is given by:
\begin{align*}
\Pr\!\left[~ v \randassign \sigma_{\wre}(x) \,:\, 
H(v) = \ell ~\Big|~ C(v) = \ell ~\right] \in I
{,}
\end{align*}
which can be written as:
\begin{align*}
\Pr\!\left[~ s \randassign \wre \,:\, 
s \models h_\ell(x) ~\Big|~ s \models \psi_\ell(x) ~\right] \in I
{.}
\end{align*}
By using \StatEL{}, this can be formalized as:
\begin{align}
\M, \wre \models \Precision_{\ell,I}(x)
~\mbox{ where }~
\Precision_{\ell,I}(x) \eqdef \psi_\ell(x) \erightarrow \PR{I} h_\ell(x)
{.}
\end{align}
Note that the precision depends on the test data sampled from the distribution $\sigma_{\wre}(x)$, hence on the real world $\wre$ in which we are located.
Hence the measurement variable $x$ in $\Precision_{\ell,I}(x)$ is interpreted using the stochastic assignment $\sigma_{\wre}$ in the world $\wre$.

Symmetrically, \emph{recall} (\emph{true positive rate}) is defined as the conditional probability that the prediction is correct given that the actual class is positive; i.e., 
${\it recall} = \frac{\TP}{\TP + \FN}$.
Then the recall being within $I$ is formalized as:
\begin{align}
\Recall_{\ell,I}(x) \eqdef h_\ell(x) \erightarrow \PR{I} \psi_\ell(x)
{.}
\end{align}
In Table~\ref{table:confusion} we show the formalization of other notions of classification performance using \StatEL{}.

\begin{figure}[t]
\centering
\begin{tikzpicture}
\coordinate (W0) at (-1.5,2.8) node at (W0) [right] {{\scriptsize Real world $\wre$}};
\draw [black, dotted, name path=rectangle, rotate=0] (1.0,2.0) rectangle +(7.1,3.1);
\coordinate (W1) at (-1.5,0.8) node at (W1) [right] {{\scriptsize Possible world $w'$}};
\draw [black, dotted, name path=rectangle, rotate=0] (1.0,0.05) rectangle +(7.1,1.5);
\coordinate (Dataset) at (1.8,3.5) node at (Dataset) {{\scriptsize dataset}};
\filldraw [gray!15] (1.8,2.5) circle [x radius=6mm, y radius=2mm, rotate=0];
\draw [black, very thin] (1.8,2.5) circle [x radius=6mm, y radius=2mm, rotate=0];
\fill [gray!15!white] (1.2,2.5) rectangle (2.4,3);
\filldraw [gray!15] (1.8,3) circle [x radius=6mm, y radius=2mm, rotate=0];
\draw [black, very thin] (1.8,3) circle [x radius=6mm, y radius=2mm, rotate=0];
\draw [black, very thin] (1.2,2.5) -- (1.2,3.0);
\draw [black, very thin] (2.4,2.5) -- (2.4,3.0);
\coordinate (DB) at (1.83,2.6) node at (DB) {$\sigma_{\!\wre}(x)$};
\filldraw [gray!15] (1.8,0.5) circle [x radius=6mm, y radius=2mm, rotate=0];
\draw [black, very thin] (1.8,0.5) circle [x radius=6mm, y radius=2mm, rotate=0];
\fill [gray!15!white] (1.2,0.5) rectangle (2.4,1);
\filldraw [gray!15] (1.8,1) circle [x radius=6mm, y radius=2mm, rotate=0];
\draw [black, very thin] (1.8,1) circle [x radius=6mm, y radius=2mm, rotate=0];
\draw [black, very thin] (1.2,0.5) -- (1.2,1.0);
\draw [black, very thin] (2.4,0.5) -- (2.4,1.0);
\coordinate (DB2) at (1.83,0.6) node at (DB2) {$\sigma_{\!w'}(x)$};
\filldraw [gray!5!white, name path=rectangle, rotate=0] (5.2,3.8) rectangle +(1.5,1);
\draw [name path=rectangle, rotate=0] (5.2,3.8) rectangle +(1.5,1);
\coordinate (HI1) at (4.2,3.1) node at (HI1) {};
\coordinate (HI2) at (5.1,4.3) node at (HI2) {};
\coordinate (HL1) at (5.95,4.45) node at (HL1) [above] {{\scriptsize Oracle}};
\coordinate (HL2) at (5.95,4.17) node at (HL2) [above] {{\tiny (human)}};
\coordinate (H) at (6.2,4.07) node at (H) [left] {$H$};
\coordinate (Hout) at (6.8,4.3) node at (Hout) {};
\coordinate (HO) at (7.4,4.3) node at (HO) [right] {~$\ell$};
\draw [->] (HI1)--(HI2);
\draw [->] (Hout)--(HO);
\coordinate (INPUT1) at (4.25,4.05) node at (INPUT1) [above] {{\tiny input}};
\coordinate (OUTPUT1) at (7.2,4.5) node at (OUTPUT1) [above] {{\tiny output}};
\filldraw [gray!5!white, name path=rectangle, rotate=0] (5.2,2.3) rectangle +(1.5,1);
\draw [name path=rectangle, rotate=0] (5.2,2.3) rectangle +(1.5,1);
\coordinate (SL) at (2.55,2.8) node at (SL) {};
\coordinate (S) at (2.7,3.15) node at (S) {~~~~~~{\tiny sampling}};
\coordinate (SR) at (3.45,2.8) node at (SR) {};
\coordinate (I) at (4.5,2.8) node at (I) [left] {$\sigma_{\!s}(x)$};
\coordinate (CI) at (5.1,2.8) node at (CI) {};
\coordinate (L) at (5.95,2.85) node at (L) [above] {{\scriptsize Classifier}};
\coordinate (C) at (6.2,2.65) node at (C) [left] {$C$};
\coordinate (CO) at (6.8,2.8) node at (CO) {};
\coordinate (O) at (7.4,2.8) node at (O) [right] {~$\ell$};
\draw [->, densely dashed] (SL)--(SR);
\draw [->] (I)--(CI);
\draw [->] (CO)--(O);
\filldraw [gray!5!white, name path=rectangle, rotate=0] (5.2,0.3) rectangle +(1.5,1);
\draw [name path=rectangle, rotate=0] (5.2,0.3) rectangle +(1.5,1);
\coordinate (SL2) at (2.55,0.8) node at (SL2) {};
\coordinate (S2) at (2.7,1.15) node at (S2) {~~~~~~{\tiny sampling}};
\coordinate (SR2) at (3.45,0.8) node at (SR2) {};
\coordinate (I2) at (4.55,0.8) node at (I2) [left] {$\sigma_{\!s'}(x)$};
\coordinate (CI2) at (5.1,0.8) node at (CI2) {};
\coordinate (L2) at (5.95,0.85) node at (L2) [above] {{\scriptsize Classifier}};
\coordinate (C2) at (6.2,0.65) node at (C2) [left] {$C$};
\coordinate (CO2) at (6.8,0.8) node at (CO2) {};
\coordinate (O2) at (7.4,0.8) node at (O2) [right] {~$\ell$};
\draw [->, densely dashed] (SL2)--(SR2);
\draw [->] (I2)--(CI2);
\draw [->] (CO2)--(O2);
\coordinate (CompD0) at (1.9,2.15) node at (CompD0) {};
\coordinate (CompD1) at (1.9,1.35) node at (CompD1) {};
\draw [<->,double] (CompD0)--(CompD1);
\coordinate (CompL0) at (7.7,2.45) node at (CompL0) {};
\coordinate (CompL1) at (7.7,1.1) node at (CompL1) {};
\coordinate (RelationL) at (1.5,1.77) node at (RelationL) {$\Repsd$};
\coordinate (Robust) at (8.6,1.78) node at (Robust) {{\scriptsize Robustness}};
\draw [<->,double] (CompL0)--(CompL1);
\coordinate (CompHC0) at (7.7,3.9) node at (CompHC0) {};
\coordinate (CompHC1) at (7.7,3.25) node at (CompHC1) {};
\coordinate (Performance) at (8.7,3.6) node at (Performance) {{\scriptsize Performance}};
\draw [<->,double] (CompHC0)--(CompHC1);
\end{tikzpicture}
\caption{
The classification performance compares the conditional probability of the human $H$'s output with that by the classifier $C$'s.
On the other hand, the robustness compares the conditional probability in the real world $\wre$ with that in a possible world $w'$ that is close to $\wre$ in terms of $\Repsd$.
Note that an adversary's choice of the test dataset $\sigma_{w'}(x)$ is formalized by the non-deterministic choice of the possible world $w'$.
\label{fig:performance-robustness}}
\end{figure}
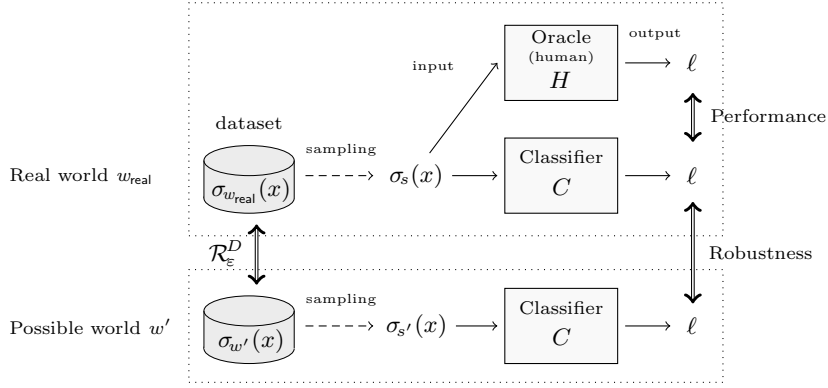

\section{Formalizing the Robustness of Classifiers}
\label{sec:ML:robustness}

Many studies have found attacks on the robustness of statistical machine learning~\cite{Chakraborty:18:arxiv}.
An input data that violates the robustness of classifiers is called an \emph{adversarial example}~\cite{Szegedy:14:ICLR}.
It is designed to make a classifier fail to predict the actual class $\ell$, but is recognized to belong to $\ell$ from human eyes.
For example, in computer vision, Goodfellow et al.~\cite{Goodfellow:ICLR:15} create an image by adding undetectable noise to a panda's photo so that 
humans can still recognize the perturbed image as a panda, but a classifier 
misclassifies it as a gibbon.

In this section we formalize robustness notions for classifiers by using epistemic operators in \StatEL{} (See Fig.~\ref{fig:performance-robustness} for an overview of the formalization).
In addition, we present some relationships between classification performance and robustness, which suggest robustness-related properties that have not been formalized in the literature as far as we know.

\subsection{Total Correctness of Classifiers}
\label{sub:security-classifiers}

We first note that the \emph{total correctness} of classifiers could be formalize as a classification performance (e.g., precision, recall, or accuracy) in the presence of all possible inputs from adversaries.
For example, the total correctness could be formalized as $\M \models \Recall_{\ell,I}(x)$, which represents that $\Recall_{\ell,I}(x)$ is satisfies in all possible worlds of $\M$.

In practice, however, it is not possible or tractable to check whether the classification performance is achieved for all possible dataset and for all possible inputs,
e.g., when $\calw$ is an infinite set.
Hence we need a weaker form of correctness notions, which may be tested in a certain way.
In the following sections, we deal with robustness notions that are weaker than total correctness.

\subsection{Probabilistic Robustness against Targeted Attacks}
\label{sub:target-robustness}

When a robustness attack aims at misclassifying an input as a specific target label, then it is called a \emph{targeted attack}.
For instance, in the above-mentioned attack by~\cite{Goodfellow:ICLR:15}, a gibbon is the target into which a panda's photo is misclassified.

To formalize the robustness, let $\Repsd \subseteq \calw \times \calw$ be an accessibility relation that relates two worlds having closer inputs, i.e.,
\[
\Repsd \eqdef
\left\{ (w, w') \in \calw\times\calw \,\mid\, D(\sigma_{w}(x) \parallel \sigma_{w'}(x)) \le \varepsilon
\right\},
\]
where $D$ is some divergence or distance.
Intuitively, $ (w, w') \in \Repsd$ implies that the two distributions $\sigma_{w}(x)$ and $\sigma_{w'}(x)$ of inputs to the classifier $C$ represent close datasets in terms of~$D$ (e.g., two test datasets consisting of slightly different images that look pandas from the human' eyes). 
Then an epistemic formula $\MKeD \phi$ represents that we are confident that $\phi$ is true as far as the classifier $C$ classifies the test data that are perturbed by noise of a level $\varepsilon$ or smaller\footnote{This usage of modality relies on the fact that the value of the measurement variable $x$ can be different in different possible worlds.}.

Now we discuss how we formalize robustness using the epistemic operator $\MKeD$ as follows.
A first definition of robustness against targeted attacks might be:
\[
\M, \wre \models h_{\sf panda}(x) \erightarrow 
\MKeD \PR{0} \psi_{\sf gibbon}(x),
\]
which represents that a panda's photo $x$ will not be recognized as a gibbon at all after the photo is perturbed by noise.
However, this does not express probability or cover the case where the human cannot recognize the perturbed image as a panda, for example, when the image is perturbed by a transformation such as linear displacement, rescaling and rotation~\cite{Athalye:18:ICML}.
Instead, for some $\delta\in[0, 1]$, we formalize a notion of \emph{probabilistic robustness against targeted attacks} by:
\begin{align*}
\TargetRobust_{{\sf panda}, \delta}(x, {\sf gibbon}) \eqdef
\MKeD \bigl( h_{\sf panda}(x) \erightarrow \PR{[0, \delta]} \psi_{\sf gibbon}(x) \bigr).
\end{align*}

Since $L^p$-norms are often regarded as reasonable approximations of human perceptual distances~\cite{Carlini17:SP}, they are used as distance constraints on the perturbation in many researches on targeted attacks (e.g.~\cite{Szegedy:14:ICLR,Goodfellow:ICLR:15,Carlini17:SP}).
To represent the robustness against these attacks in our model, we should take the metric $D$ 
as the $\infty$-Wasserstein distance $\Winfu$ ( in terms of the $L^p$ metric $d$) between the two distributions $\sigma_{w}(x)$ and $\sigma_{w'}(x)$
\footnote{A coupling that achieves $\Winfu(\sigma_{w}(x), \sigma_{w'}(x)) \le \varepsilon$ provides a transformation of an image in $\supp(\sigma_{w}(x))$ to another in $\supp(\sigma_{w'}(x))$ perturbed by a level $\varepsilon$ of noise.}.

\subsection{Probabilistic Robustness against Non-Targeted Attacks}
\label{sub:total-robustness}

Next we formalize \emph{non-targeted attacks}~\cite{Moosavi:16:CVPR,Madry:18:ICLR} in which adversaries try to misclassify inputs as some arbitrary incorrect labels (i.e., not as a specific label like a gibbon). 
Compared to targeted attacks, this kind of attacks are easier to mount, but harder to defend.

A notion of \emph{probabilistic robustness against non-targeted attacks} can be formalized for some $I = [1-\delta, 1]$ by:
\begin{align}
\TotalRobust_{\ell, I}(x) 
&\eqdef
\MKeD \bigl( h_\ell(x) \erightarrow \PR{I} \psi_\ell(x) \bigr)
= \label{eq:robust-recall}
\MKeD \Recall_{\ell, I}(x)
{.}
\end{align}
Then we derive that $\TotalRobust_{{\sf panda}, I}(x)$ implies $\TargetRobust_{{\sf panda}, \delta}(x, {\sf gibbon})$,
namely, robustness against non-targeted attacks is not weaker than robustness against targeted attacks.

Next we note that by \eqref{eq:robust-recall}, robustness can be regarded as recall in the presence of perturbed noise.
This implies that for each property $\phi$ in Table~\ref{table:confusion}, we could consider $\MKeD \phi$ as a property related to robustness although these have not been formalized in the literature of robustness of machine learning as far as we recognize.
For example, $\MKeD \Precision_{\ell,i}(x)$ represents that in the presence of perturbed noise, the prediction is correct with a probability $i$ given that it is positive.
For another example, $\MKeD \Accuracy_{\ell,i}(x)$ represents that in the presence of perturbed noise, the prediction is correct (whether it is positive or negative) with a probability $i$.

Finally, note that by the reflexivity of $\Repsd$,\, 
$\M, \wre \models \MKeD \Recall_{\ell, I}(x)$ implies $\M, \wre \models \Recall_{\ell, I}(x)$,
i.e., robustness implies recall without perturbation noise.

\section{Formalizing the Fairness of Classifiers}
\label{secML:fairness}

There have been researches on various notions of fairness in machine learning.
In this section, we formalize a few notions of fairness of classifiers by using \StatEL{}.
Here we focus on the fairness that should be maintained in the \emph{impact}, i.e., the results of classification, rather than the \emph{treatment}\footnote{For instance, \emph{fairness through unawareness} requires that protected attributes (e.g., race, religion, or gender) are not explicitly used in the prediction process.
However, \StatEL{} may not be suited to formalizing such a property in treatment.}.

To formalize fairness notions, we use a distributional Kripke model $\M =(\calw, (\calr_a)_{a\in\cala}, \allowbreak (V_s)_{s\in\cals})$ where $\calw$ includes a possible world $w_d$ having a dataset $d$ from which an input to the classifier $C$ is drawn.
Recall that $x$, $y$, and $\hy$ are measurement variables denoting the input to the classifier $C$, the actual class label, and the predicted label by $C$, respectively.
In each world $w$,\, $\sigma_{w}(x)$ is the distribution of $C$'s input over $\cald$, (i.e., the test data distribution),
$\sigma_{w}(y)$ is the distribution of the actual label over $\Label$, and
$\sigma_{w}(\hy)$ is the distribution of $C$'s output over $\Label$.
For each group $G\subseteq\cald$ of inputs, we introduce a static formula $\eta_{G}(x)$ representing that an input $x$ belongs to~$G$.
We also introduce a formula $\xi_d$ representing that all data are drawn from some subset of the dataset $d$.
Formally, these are interpreted by:
\begin{itemize}
\item For each state $s\in\cals$,\, $s\models \eta_{G}(x)$ iff $\sigma_{s}(x) \in G$;
\item For each world $w\in\calw$,\, $w\models \xi_d$ iff 
there exists a $\cals' \subseteq\cals$ s.t. 
$w[s] = \frac{w_d[s]}{\sum_{s' \in \cals'} w_d[s']}$ if $s \in \cals'$, and $w[s] = 0$ otherwise.
\end{itemize}
For two worlds $w$ and $w'$, we write $w \models \Qw{w'} \psi$ to denote that $w \models \PR{1} \psi$ and $s \not\models \psi$ for all $s\in\supp(w')\setminus\supp(w)$.

Then we obtain the following proposition on conditional indistinguishability.

\begin{restatable}[Conditional indistinguishability in a world $w_d$]{prop}{CondIndLocal}
\label{prop:CondIndLocal}
Let $\M =(\calw, (\calr_a)_{a\in\cala}, \allowbreak (V_s)_{s\in\cals})$ 
be a distributional Kripke model with the universe $\calw = \Dists\cals$.
Let $w_d$ be a world with a dataset $d$,
$\psi$ and $\psi'$ be static formulas, and $a\in\cala$.
\begin{enumerate}
\item[(i)] 
$\M, w_d \models \psi \erightarrow \neg \ov{\MPa } \bigl( \xi_d \wedge \Qw{w_d} \psi' \bigr)$
iff for any $w, w' \in \calw$,\, 
$\M, w \models \xi_d \wedge \Qw{w_d} \psi$ and $\M, w' \models \xi_d \wedge \Qw{w_d} \psi'$
imply $(w, w') \in \Ra$.
\item[(ii)] 
If $\calr_a$ is symmetric, then $\M, w_d \models \psi \erightarrow \neg \ov{\MPa }\Qw{w_d} \psi'$ iff $\M, w_d \models \psi' \erightarrow \neg \ov{\MPa }\Qw{w_d} \psi$.
\end{enumerate}
\end{restatable}
See Appendix~\ref{sec:proof:prop:CondInd} for the proof.

Now we formalize three popular notions of fairness of classifiers by using counterfactual epistemic operators (introduced in Section~\ref{sec:counterfactual}) as follows.

\subsection{Group Fairness (Statistical Parity)}
\label{sub:group-fairness}
The \emph{group fairness} formulated as \emph{statistical parity}~\cite{Dwork:12:ITCS} is the property that the output distributions of the classifier are identical for different groups.
Formally, for each $b = 0, 1$ and a group $G_b \subseteq \cald$, let $\mu_{G_b}$ be the distribution of the output (over $\Label$) of the classifier $C$ when the input is sampled from a dataset $d$ and belongs to $G_b$.
Then the statistical parity up to bias $\varepsilon$ is formalized using the total variation $\Dtv$ by
$\Dtv( \mu_{G_0} \| \mu_{G_1} ) \leq \varepsilon$.

To express this using \StatEL{}, 
we define an accessibility relation $\Repstv$ in $\M$~by: 
\begin{align}\label{eq:Repstv}
\Repstv &\eqdef 
\left\{ (w, w') \in \calw\times\calw \mid 
\tvdiverge{\sigma_{w}(\hy)\!}{\!\sigma_{w'}(\hy)} 
\le \varepsilon
\right\}
{.}
\end{align}
Intuitively, $(w, w') \in \Repstv$ represents that the two probability distributions $\sigma_{w}(\hy)$ and $\sigma_{w'}(\hy)$ of the outputs by the classifier $C$ respectively in $w$ and in $w'$ are close in terms of $\Dtv$.
Note that $\sigma_{w}(\hy)$ and $\sigma_{w'}(\hy)$ respectively represent $\mu_{G_0}$ and $\mu_{G_1}$.

Then the statistical parity w.r.t. groups $G_0, G_1$ means that in terms of $\Repstv$, we cannot distinguish a world having a dataset $d$ and satisfying $\eta_{G_0}(x) \wedge \psi(x, \hy)$ from another satisfying $\eta_{G_1}(x) \wedge \psi(x, \hy)$.
By Proposition~\ref{prop:CondIndLocal}, this is expressed as:
\[
\M, w_d \models \GrpFair(x, \hy)
\]
where
$
\GrpFair(x, \hy) \eqdef
\bigl( \eta_{G_0}(x) \wedge \psi(x, \hy) \bigr) \erightarrow
\neg \ov{\MPetv} \bigl( \xi_d \wedge \Qw{w_d} (\eta_{G_1}(x) \wedge \psi(x, \hy)) \bigr)
$.

\subsection{Individual Fairness (as Lipschitz Property)}
\label{sub:individual-fairness}
The \emph{individual fairness} formulated as a Lipschitz property~\cite{Dwork:12:ITCS} is the property that the classifier outputs similar labels given similar inputs.
Formally, for $v, v' \in \cald$, let $\mu_{v}$ and $\mu_{v'}$ be the distributions of the outputs (over $\Label$) of the classifier $C$ when the inputs are $v$ and $v'$, respectively.
Then the individual fairness is formalized using a divergence $D: \Dists\Label\times\Dists\Label \rightarrow \realsnng$, a metric $r: \cald\times\cald \rightarrow \realsnng$, and a threshold $\varepsilon\in\realsnng$ by
$\diverge{\mu_{v}}{\mu_{v'}} \leq \varepsilon\cdot r(v, v')$.

To express this using \StatEL{}, we define an accessibility relation $\Repslps$ in $\M$ for the metric $r$ and the divergence $D$ as follows:
\begin{align}\label{eq:Repslps}
\Repslps &\eqdef 
\left\{ (w, w') \in \calw\times\calw ~\Big|~ 
\begin{array}{ll}
v \in\supp(\sigma_{w}(x)),~
v' \in\supp(\sigma_{w'}(x)), \\
\diverge{\sigma_{w}(\hy)\!}{\!\sigma_{w'}(\hy)} 
\le \varepsilon \cdot 
r(v, v')
\end{array}
\right\}
{.}
\end{align}
Intuitively, $(w, w') \in \Repslps$ represents that, when inputs are closer in terms of the metric $r$, the classifier $C$ outputs closer labels in terms of the divergence $D$.

Then the individual fairness w.r.t. $r$ and $D$ means that in terms of $\Repslps$, we cannot distinguish between the two worlds $w$ and $w'$ where $\psi(x, \hy)$ is satisfied (i.e., $C$ outputs $\hy$ given an input $x$).
By Proposition~\ref{prop:CondIndLocal}, this is expressed as:
\[
\M, w_d \models \IndFair(x, \hy)
\]
where
$\IndFair(x, \hy) \eqdef
\psi(x, \hy) \erightarrow
\neg \ov{\MPelps} \bigl( \xi_d \wedge \Qw{w_d} \psi(x, \hy) \bigr)$.

This represents that by observing the classifier's output $\hy$, we can less distinguish two worlds $w$ and $w'$ 
when their inputs $\sigma_{w}(x)$ and $\sigma_{w'}(x)$ are closer.

\subsection{Equal Opportunity}
\label{sub:equal-opportunity}

\emph{Equal opportunity}~\cite{Hardt:16:NIPS,Zafar:17:WWW} is the property that the recall (true positive rate) is the same for all the groups.
Formally, given an advantage class $\ell\in\Label$ (e.g., not defaulting on a loan) and a group $G\subseteq\cald$ of inputs with a protected attribute (e.g., race), a classifier $C$ is said to satisfy equal opportunity of $\ell$ w.r.t. $G$ 
if it holds for each $\hell\in\Label$ that:
\begin{align}\label{eq:equal-opportunity}
\Pr[ C(x) = \hell \mid x\in G,\, H(x) = \ell ] =
\Pr[ C(x) = \hell \mid x\in \cald\setminus\!G,\, H(x) = \ell ].
\end{align}

If we allow the logic to use the universal quantification over the probability value $i$, then 
the case of $\hell = \ell$ in \eqref{eq:equal-opportunity} could be expressed as:
\[
\forall i \in [0, 1].~
\bigl(
\xi_d \wedge \eta_{G}(x) \erightarrow \Recall_{\ell,i}(x)
\bigr)
\leftrightarrow
\bigl(
\xi_d \wedge \neg\eta_{G}(x) \erightarrow \Recall_{\ell,i}(x)
\bigr)
{.}
\]
However, instead of allowing for this universal quantification, 
we can use the modal operators $\ov{\MPetv}$ (defined by~\eqref{eq:Repstv}) with $\varepsilon = 0$,
and represent equal opportunity as the fact that we cannot distinguish a world having a dataset $d$ and satisfying $\eta_{G}(x) \wedge \psi(x, \hy) \wedge h_{\ell}(x)$ from another satisfying $\neg\eta_{G}(x) \wedge \psi(x, \hy) \wedge h_{\ell}(x)$ as follows:
\[
\EqOpp(x, \hy) \!\eqdef\!
\bigl( \eta_{G}(x) \wedge \psi(x, \hy) \wedge h_{\ell}(x) \bigr) \erightarrow
\neg \ov{\MPztv} \bigl( \xi_d \wedge \Qw{w_d} (\neg\eta_{G}(x) \wedge \psi(x, \hy) \wedge h_{\ell}(x)) \bigr)
{.}
\]

\section{Related Work}
\label{sec:related}
In this section, we provide a brief overview of related work on the specification of statistical machine learning and on epistemic logic for describing specification.

\paragraph{Desirable properties of statistical machine learning.}

There have been a large number of papers on attacks and defences for deep neural networks~\cite{Szegedy:14:ICLR,Chakraborty:18:arxiv}.
Compared to them, however, not much work has been done to explore the formal specification of various properties of machine learning.
Seshia et al.~\cite{Seshia:18:ATVA} present a list of desirable properties of DNNs (deep neural networks) although most of the properties are presented informally without mathematical formulas.
As for robustness, Dreossi et al.~\cite{Dreossi:19:VNN}
propose a unifying formalization of adversarial input generation in a rigorous and organized manner, although they formalize and classify attacks (as optimization problems) rather than define the robustness notions themselves.
Concerning the fairness notions, 
Gajane~\cite{Gajane:17:arxiv} surveys the formalization of fairness notions for machine learning and present some justification based on social science literature.

\paragraph{Epistemic logic for describing specification.}
Epistemic logic~\cite{vonWright:51:book} has been studied to represent and reason about knowledge~\cite{Fagin:95:book,Halpern:03:book,Halpern:05:JCS}, and has been applied to describe various properties of systems.

The \emph{BAN logic}~\cite{Burrows:90:TOCS}, proposed by Burrows, Abadi and Needham, is a notable example of epistemic logic used to model and verify the authentication in cryptographic protocols.
To improve the formalization of protocols' behaviours, some epistemic approaches integrate process calculi~\cite{Hughes:04:JCS,Dechesne:07:LPAR,Chadha:09:Forte}.

Epistemic logic has also been used to formalize and reason about privacy properties, including anonymity~\cite{Syverson:99:FM,Halpern:05:JCS,Garcia:05:FMSE,Kawamoto:07:JSIAM},
receipt-freeness of electronic voting protocols~\cite{Jonker:06:WOTE},
and privacy policy for social network services~\cite{Pardo:14:SEFM}.
Temporal epistemic logic is used to express information flow security policies~\cite{Balliu:11:PLAS}.

Concerning the formalization of fairness notions, previous work in formal methods has modeled different kinds of fairness involving timing by using temporal logic rather than epistemic logic.
As far as we know, no previous work has formalized fairness notions of machine learning using counterfactual epistemic operators.

\paragraph{Formalization of statistical properties.}

In studies of philosophical logic, Lewis~\cite{Lewis:80:subjectivist} shows the idea that when a random value has various possible probability distributions, then those distributions should be represented on distinct possible worlds.
Bana~\cite{Bana:17:EPSP} puts Lewis's idea in a mathematically rigorous setting. 
Recently, a modal logic called statistical epistemic logic~\cite{Kawamoto:19:FC:arxiv} is proposed and is used to formalize statistical hypothesis testing and the notion of differential privacy~\cite{Dwork:06:ICALP}.
Independently of that work, French et al.~\cite{French:19:AAAMAS} propose a probability model for a dynamic epistemic logic in which each world is associated with a subjective probability distribution over the universe, without dealing with non-deterministic inputs or statistical divergence.

\section{Conclusion}
\label{sec:conclude}
We have shown a logical approach to formalizing statistical classifiers and their desirable properties in a simple and abstract way.
Specifically, we have introduced a formal model for probabilistic behaviours of classifiers and non-deterministic adversarial inputs using a distributional Kripke model.
Then we have formalized the classification performance, robustness, and fairness of classifiers by using \StatEL{}.
Moreover, we have also clarified some relationships among properties of classifiers, and relevance between classification performance and robustness.
To formalize fairness notions, we have introduced a notion of counterfactual knowledge and shown some techniques to express conditional indistinguishability.
As far as we know, this is the first work that uses logical formulas to express statistical properties of machine learning, and that provides epistemic (resp. counterfactually epistemic) views on robustness (resp. fairness) of classifiers.

In future work, we are planning to include temporal operators in the specification language and to formally reason about system-level properties of learning-based systems.
We are also interested in developing a general framework for the formal specification of machine learning associated with testing methods and possibly extended with Bayesian networks.
Our future work also includes an extension of \StatEL{} to formalize machine learning other than classification problems.
Another possible direction of future work would be to clarify the relationships between our counterfactual epistemic operators and more general notions of counterfactual knowledge in previous work such as~\cite{Williamson:07:GPS}.

\appendix
\section{Proofs for Propositions~\ref{prop:CondInd} and~\ref{prop:CondIndLocal}}
\label{sec:proof:prop:CondInd}

\CondInd*
\begin{proof}
We first prove the claim (i) as follows.
We show the direction from left to right.
Assume that $\M \models \psi \erightarrow \neg \ov{\MPa }\PR{1} \psi'$.
Let $w, w' \in \calw$ satisfy $\M, w \models \PR{1} \psi$ and $\M, w' \models \PR{1} \psi'$.
Then 
$w|_{\psi} = w$.
By $\M, w \models  \psi \erightarrow \neg \ov{\MPa }\PR{1} \psi'$, 
we obtain $\M, w|_{\psi} \models \neg\ov{\MPa} \PR{1} \psi'$,
which is logically equivalent to $\M, w|_{\psi} \models \ov{\MKa} \neg \PR{1} \psi'$.
By the definition of $\ov{\MKa}$, for every $w''\in\calw$,\, $\M, w'' \models \PR{1} \psi'$ implies $(w|_{\psi}, w'')\in\Ra$.
Then, since $w|_{\psi} = w$ and $\M, w' \models \PR{1} \psi'$, we obtain $(w, w')\in\Ra$.

Next we show the other direction as follows.
Assume the right hand side.
Let $w \in \calw$ such that $\M, w \models \PR{1} \psi$.
Then for every $w' \in \calw$,\, $\M, w' \models \PR{1} \psi'$ implies $(w, w') \in \Ra$.
By the definition of $\ov{\MKa}$, we have $\M, w \models \ov{\MKa} \neg \PR{1} \psi'$, which is equivalent to $\M, w \models \neg\ov{\MPa} \PR{1} \psi'$.
By $\M, w \models \PR{1} \psi$, we have $w|_{\psi} = w$, hence $\M, w|_{\psi} \models \neg\ov{\MPa} \PR{1} \psi'$.
Therefore $\M, w \models  \psi \erightarrow \neg \ov{\MPa }\PR{1} \psi'$.

Finally, the claim (ii) follows from the claim (i) immediately.
\qed
\end{proof}

\CondIndLocal*
\begin{proof}
We first prove the claim (i) as follows.
We show the direction from left to right.
Assume that $\M, w_d \models \psi \erightarrow \neg \ov{\MPa }\bigl( \xi_d \wedge \Qw{w_d} \psi' \bigr)$.
Let $w, w' \in \calw$ satisfy $\M, w \models \xi_d \wedge \Qw{w_d} \psi$ and $\M, w' \models \xi_d \wedge \Qw{w_d} \psi'$.
Then $w_d|_{\psi} = w$ and $w_d|_{\psi'} = w'$.
By $\M, w_d \models  \psi \erightarrow \neg \ov{\MPa } \bigl( \xi_d \wedge \Qw{w_d} \psi' \bigr)$, 
we obtain $\M, w_d|_{\psi} \models \neg\ov{\MPa} \bigl( \xi_d \wedge \Qw{w_d} \psi' \bigr)$,
which is logically equivalent to $\M, w_d|_{\psi} \models \ov{\MKa} \neg \bigl( \xi_d \wedge \Qw{w_d} \psi' \bigr)$.
By the definition of $\ov{\MKa}$ and $\M, w' \models \xi_d \wedge \Qw{w_d} \psi'$,
we have $(w_d|_{\psi}, w')\in\Ra$.
Therefore, by $w|_{\psi} = w$, we obtain $(w, w')\in\Ra$.

Next we show the other direction as follows.
Assume the right hand side.
Let $w \in \calw$ such that $\M, w \models \xi_d \wedge \Qw{w_d} \psi$.
Then for every $w' \in \calw$,\, $\M, w' \models \xi_d \wedge \Qw{w_d} \psi'$ implies $(w, w') \in \Ra$.
By the definition of $\ov{\MKa}$, we have $\M, w \models \ov{\MKa} \neg \bigl(\xi_d \wedge \Qw{w_d} \psi' \bigr)$, which is equivalent to $\M, w \models \neg\ov{\MPa} \bigl(\xi_d \wedge \Qw{w_d} \psi'\bigr)$.
By $\M, w \models \xi_d \wedge \Qw{w_d} \psi$, we have $w_d|_{\psi} = w$, hence $\M, w_d|_{\psi} \models \neg\ov{\MPa} \bigl( \xi_d \wedge \Qw{w_d} \psi' \bigr)$.
Therefore $\M, w_d \models  \psi \erightarrow \neg \ov{\MPa } \bigl( \xi_d \wedge \Qw{w_d} \psi' \bigr)$.

Finally, the claim (ii) follows from the claim (i) immediately.
\qed
\end{proof}

 \bibliographystyle{splncs04}
 \bibliography{short,short-ML}

\end{document}